\RequirePackage{amsmath}
\documentclass[runningheads]{llncs}
\usepackage[margin=1in]{geometry}
\usepackage{graphicx, hyperref, color}
\usepackage{amsfonts}
\usepackage{algorithm}
\usepackage{algpseudocode}

\usepackage{comment}
\usepackage{bm, bbm}

\newcommand{\argmax}{\mathop{\rm arg~max}\limits}

\newcommand{\atime}{\mathrm{atime}}
\newcommand{\pos}{\mathrm{pos}}
\newcommand{\optcost}{\mathrm{opt\text{-}cost}}
\newcommand{\sur}{\mathrm{sur}}
\newcommand{\free}{\mathrm{free}}

\newcommand{\dD}{d_{\mathrm{D}}}

\newcounter{mpproblem}[section]

\makeatletter
\newenvironment{mpproblem}[1]%
{%
    \protected@edef\@currentlabelname{#1}%
    \par\vspace{\baselineskip}\noindent%
    \ifx#1\empty %
    \else \refstepcounter{mpproblem}$($#1$)$ %
    \fi%
    \hfill%
    $\left|%
    \hfill%
    \hspace{0.00\textwidth}%
    \@fleqntrue\@mathmargin\parindent%
    \begin{minipage}{0.86\textwidth}%
    \vspace{-1.0\baselineskip}%
}%
{%
    \end{minipage}%
    \@fleqnfalse%
    \right.$%
    \par\vspace{\baselineskip}\noindent%
    \ignorespacesafterend%
}%
\makeatother

\begin{document}
\title{Deterministic Primal-Dual Algorithms for Online $k$-way Matching with Delays\thanks{This work was partly supported by JSPS KAKENHI Grant Numbers 20H05795, 22H05001, and 21H03397.}}
\titlerunning{A Primal-Dual Algorithm for $k$-Way Matching with Delays}
%
\author{Naonori Kakimura\inst{1}\orcidID{0000-0002-3918-3479} \and
Tomohiro Nakayoshi\inst{2}}
\authorrunning{N. Kakimura and T. Nakayoshi}
%
\institute{Keio University, 3-14-1 Hiyoshi, Kohoku-ku, Yokohama, Kanagawa 223-8522, Japan
\email{kakimura@math.keio.ac.jp}\\ \and
The University of Tokyo, 7-3-1 Hongo, Bunkyo-ku, Tokyo 113-8654, Japan
\email{nakayoshi-tomohiro@g.ecc.u-tokyo.ac.jp}}

\maketitle              
\begin{abstract}
In this paper, we study the Min-cost Perfect $k$-way Matching with Delays~($k$-MPMD), recently introduced by Melnyk et al.
In the problem, $m$ requests arrive one-by-one over time in a metric space.
At any time, we can irrevocably make a group of $k$ requests who arrived so far, that incurs the distance cost among the $k$ requests in addition to the sum of the waiting cost for the $k$ requests.
The goal is to partition all the requests into groups of $k$ requests, minimizing the total cost.
The problem is a generalization of the min-cost perfect matching with delays~(corresponding to $2$-MPMD).
It is known that no online algorithm for $k$-MPMD can achieve a bounded competitive ratio in general, where the competitive ratio is the worst-case ratio between its performance and the offline optimal value.
On the other hand, $k$-MPMD is known to admit a randomized online algorithm with competitive ratio  $O(k^{5}\log n)$ for a certain class of $k$-point metrics called the $H$-metric, where $n$ is the size of the metric space.
In this paper, we propose a deterministic online algorithm with a competitive ratio of $O(mk^2)$ for the $k$-MPMD in $H$-metric space.
Furthermore, we show that the competitive ratio can be improved to $O(m + k^2)$ if the metric is given as a diameter on a line.
%

\end{abstract}

\section{Introduction}

Consider an online gaming platform supporting two-player games such as Chess.
In such a platform, players arrive one-by-one over time, and stay in a queue to participate in a match.
The platform then tries to suggest a suitable opponent for each player from the queue.
In order to satisfy the players, the platform aims to maximize the quality of the matched games.
Specifically, we aim to minimize the distance of the matched players~(e.g., the difference of their ratings) as well as the sum of the players' waiting time.

The above situation can be modeled as the problem called Online Matching with Delays, introduced by Emek et al.~\cite{emek_online_2016}.
In the setting, arriving requests~(or players) are embedded in a metric space so that the distance of each pair is determined.
For the Online Matching with Delays, Emek et al.~\cite{emek_online_2016} proposed a randomized algorithm with a competitive ratio of $O(\log^2 n + \log \Delta)$, where $n$ is the number of points in a metric space, and $\Delta$ is the ratio of the maximum to minimum distance between two points.
The competitive ratio was later improved to $O(\log n)$ by Azar et al.~\cite{azar_polylogarithmic_2017}.
We remark that both algorithms require that a metric space is finite and all the points in the metric space are known in advance~(we note that arriving requests may be embedded into the same point more than once).
Bienkowski et al.~\cite{bienkowski_primal-dual_2018} presented a primal-dual algorithm with a competitive ratio of $O(m)$, where $m$ is the number of requests.
Another algorithm with a better competitive ratio of $O(m^{0.59})$ was proposed by Azar et al.~\cite{azar_deterministic_2020}.

In this paper, we consider a generalization of Online Matching with Delays, called the Min-cost Perfect $k$-way Matching with Delays~($k$-MPMD)~\cite{melnyk_online_2021}.
In the problem, requests arrive one-by-one over time.
At any time, instead of choosing a pair of requests, we make a group of $k$ requests.
This corresponds to an online gaming platform that allows more than two players to participate, such as mahjong~($k = 4$), Splatoon~($k = 8$), Apex Legends~($k = 60$), and Fortnite~($k = 100$).
Then we aim to partition all the requests into groups of size-$k$ subsets, minimizing the sum of the distance of the requests in the same group and the total waiting time.

To generalize to $k$-MPMD, 
it is necessary to measure the distance of a group of $k>2$ requests.
That is, we need to introduce a metric space that defines distances for any subset of $k$ points.
Although there are many ways of generalizing a standard distance between two points to $k>2$ points in the literature~\cite{assafPartialNmetricSpaces2015,khanPossibitityNtopologicalSpaces2012}, 
Melnyk et al.~\cite{melnyk_online_2021} showed that most known generalized metrics on $k$ points cannot achieve a bounded competitive ratio for the $k$-MPMD.
Melnyk et al.~\cite{melnyk_online_2021} then introduced a new interesting class of generalized metric, called \textit{$H$-metric}, and proposed 
a randomized algorithm for the $k$-MPMD on $H$-metric with a competitive ratio of $O(k^{5}\log n)$, extending Azar et al.~\cite{azar_polylogarithmic_2017}.

The main contribution of this paper is to propose a deterministic algorithm for the $k$-MPMD on $H$-metric with a competitive ratio of $O(mk^2)$, where $m$ is the number of requests.
The proposed algorithm adopts a primal-dual algorithm based on a linear programming relaxation of the $k$-MPMD.

To design a primal-dual algorithm, we first formulate a linear programming relaxation of the offline problem, that is, when a sequence of requests is given in advance.
We remark that even the offline setting is NP-hard when $k\geq 3$, as it includes the triangle packing problem.
We first show that $H$-metric can be approximated by a standard metric~(Theorem~\ref{thm:approx_hmetric}).
This allows us to construct a linear programming problem with variables for each pair of requests such that the optimal value gives a lower bound on the offline version of the $k$-MPMD.
Using the linear programming problem, we can design a primal-dual algorithm by extending the one by Bienkowski et al.~\cite{bienkowski_primal-dual_2018} for Online Matching with Delays.
We show that, by the observation on $H$-metric~(Theorem~\ref{thm:approx_hmetric}) again, the cost of the output can be upper-bounded by the dual objective value of our linear programming problem.

An interesting special case of the $H$-metric is the diameter on a line.
That is, points are given on a 1-dimensional line, and the distance of $k$ points is defined to be the maximum difference in the $k$ points.
In the context of an online gaming platform, the diameter on a line can be interpreted as the difference of players' ratings.
In this case, we show that the competitive ratio of our algorithm can be improved to $O(m + k^2)$.
Moreover, we construct an instance such that our algorithm achieves the competitive ratio of $\Omega (m/k)$.

\subsubsection{Related Work}

An online algorithm for the matching problem was first introduced by Karp et al.~\cite{karpOptimalAlgorithmOnline1990a}.
They considered the online bipartite matching problem where arriving requests are required to match upon their arrival.
Since then, the problem has been studied extensively in theory and practice.
For example, motivated by internet advertising, Mehta et al.~\cite{mehtaAdWordsGeneralizedOnline2005a} generalized the problem to the AdWords problem.
See also Mehta~\cite{mehtaOnlineMatchingAd2013a} and Goel and Mehta~\cite{goelOnlineBudgetedMatching2008a}. 
The weighted variant of the online bipartite matching problem is considered in the literature.
It includes the vertex-weighted online bipartite matching~\cite{aggarwalOnlineVertexWeightedBipartite2011}, the problem with metric costs~\cite{bansalLog2kCompetitiveAlgorithm2007a,raghvendraRobustOptimalOnline2016,nayyarInputSensitiveOnline2017a}, and the problem with line metric cost~\cite{antoniadis2019left,fuchsOnlineMatchingLine2005,koutsoupiasOnlineMatchingProblem2004,guptaOnlineMetricMatching2012a}.
We remark that the edge-weighted online bipartite matching in general has no online algorithm with bounded competitive ratio~\cite{aggarwalOnlineVertexWeightedBipartite2011}.

This paper deals with a variant of the online matching problem with delays, in which arriving requests are allowed to make decision later with waiting costs.
Besides the related work~\cite{azar_polylogarithmic_2017,azar_deterministic_2020,bienkowski_primal-dual_2018,emek_online_2016} mentioned before, 
Liu et al.~\cite{liuImpatientOnlineMatching2018a} extended the problem to the one with non-linear waiting costs.
Other delay costs are studied in~\cite{liuOnlineMatchingConvex2022,deryckere2023online,azarMinCostMatchingConcave2021}.
Ashlagi et al.~\cite{ashlagiEdgeWeightedOnlineWindowed2023} studied the online matching problem with deadlines, where each arriving request has to make a decision by her deadline.
Pavone et al.~\cite{pavoneOnlineHypergraphMatching2020} considered online hypergraph matching with deadlines.

\subsubsection{Paper Organization}

This paper is organized as follows.
In Section~\ref{sec:prelim}, we formally define the minimum-cost perfect $k$-way matching problem and $H$-metric.
We also discuss useful properties of $H$-metrics which will be used in our analysis.
In Section~\ref{sec:algorithm}, we present our main algorithm for the $k$-MPMD on $H$-metric.
In Section~\ref{sec:lowerbound}, we show that there exists an instance such that our algorithm admits an almost tight competitive ratio.
Due to the space limitation, the proofs of lemmas and theorems are omitted, which may be found in the full version of this paper.

\section{Preliminaries}\label{sec:prelim}


\subsection{Minimum-cost Perfect $k$-way Matching with Delays}

In this section, we formally define the problem $k$-MPMD.
Let $(\chi, d)$ be a generalized metric space where $\chi$ is a set and $d: \chi^k \rightarrow [0, \infty)$ represents a distance among $k$ elements.

In the problem, $m$ requests $u_1, u_2, \ldots, u_m$ arrive one-by-one in this order.
The arrival time of $u_i$ is denoted by $\atime(u_i)$.
When $u_i$ arrives, the location $\pos(u_i)$ of $u_i$ in the metric space $\chi$ is revealed.
Thus, an instance of the problem is given as a tuple $\sigma =(V, \atime, \pos)$, where $V=\{u_1, \dots, u_m\}$,  $\atime: V\to \mathbb{R}_+$, and $\pos: V\to \chi$ such that $\atime (u_1)\leq \dots\leq \atime(u_m)$.
We note that $m$ may be unknown in advance, but we assume that $m$ is a multiple of $k$.

At any time $\tau$, with the only information for requests arrived so far, an online algorithm can make a set of $k$ requests $v_1, \ldots, v_k$ in $V$, where we say that $v_1, \ldots, v_k$ are \textit{matched}, if they satisfy the following two conditions: (a) The requests $v_1, \ldots, v_k$ have already arrived, that is, $\atime(v_i) \leq \tau$ for any $i=1,\dots, k$; (b) None of $v_1, \ldots, v_k$ has been matched to other requests yet.
The cost to match $v_1, \ldots, v_k$ at time $\tau$ is defined to be
\[
  d(\pos(v_1), \pos(v_2), \ldots, \pos(v_k)) + \sum_{i=1}^k (\tau - \atime(v_i)).
\]
The first term means the distance cost among the $k$ requests and the second term is the total waiting cost of the $k$ requests.

The objective of the problem is to design an online algorithm that matches all the requests, minimizing the total cost.
In other words, an online algorithm finds a family of disjoint subsets of size $k$ that covers all the requests.
We call a family of disjoint subsets of size $k$ a \textit{$k$-way matching}, and a $k$-way matching is called \textit{perfect} if it covers all the requests.

To measure the performance of an online algorithm, we define the competitive ratio.
For an instance $\sigma$, let $\mathcal{ALG}(\sigma)$ be the cost incurred by the online algorithm, and let $\mathcal{OPT}(\sigma)$  be the optimal cost when we know in advance a sequence of requests $V$ as well as $\atime(u_i)$ and $\pos(u_i)$ for each request $u_i$.
The \textit{competitive ratio} of the online algorithm is defined as 
$\sup_{\sigma} \frac{\mathcal{ALG(\sigma)}}{\mathcal{OPT(\sigma)}}$.

\subsection{$H$-metric}

In this section, we define $H$-metric, introduced by Melnyk et al.~\cite{melnyk_online_2021}.
Recall that a function $d:\chi^2\to[0, \infty)$ is called a \textit{distance function}~(or a \textit{metric}) if $d$ satisfies the following three axioms:
\begin{itemize}
\item \textbf{(Symmetry)} $d(p_1, p_2) = d(p_2, p_1)$ for any $p_1, p_2\in \chi$.
\item \textbf{(Positive definiteness)} $d(p_1, p_2) \geq 0$ for any $p_1, p_2\in \chi$, and $d(p_1, p_2) = 0$ if and only if $p_1 = p_2$.
\item \textbf{(Triangle inequality)} $d(p_1, p_3) \leq d(p_1, p_2) + d(p_2, p_3)$ for any $p_1, p_2, p_3\in \chi$.
\end{itemize}

We first define a $k$-point metric as a $k$-variable function satisfying generalizations of the symmetry axiom and the positive definiteness axiom.
\begin{definition}
 We call a function $d: \chi^k \rightarrow [0, \infty)$ a \emph{$k$-point metric}
  if it satisfies the following two axioms.
\begin{itemize}
  \item[$\Pi$:] For any permutation $\pi$ of $\{p_1, \ldots, p_k\}$, we have 
    $d(p_1, \ldots, p_k) = d(\pi(p_1), \ldots, \pi(p_k))$.
  \item[$O_D$:] It holds that $d(p_1, \ldots, p_k) \geq 0$. Moreover, $d(p_1, \ldots, p_k) = 0$ if and only if $p_1 = p_2 = \cdots = p_k$.
\end{itemize}
\end{definition}

There are several ways of generalizing the triangle inequality to $k$-variable functions.
One possibility is the following axiom: for any $p_1, \ldots, p_k, a \in \chi$ and any $i \in \{1, \dots, k\}$, it holds that
     \[
     \text{$\Delta_H$:\ }d(p_1, \ldots, p_k) \leq d(p_1, \ldots, p_i, \underbrace{a, \ldots, a}_{k - i}) + d(\underbrace{a, \ldots, a}_{i}, p_{i + 1}, \ldots, p_k).
     \]
We note that it is identical to the triangle inequality when $k=2$.

For a multiset $S$ on $\chi$, we denote by $elem(S)$ the set of all distinct elements contained in $S$.
In addition to the generalized triangle inequality, we consider the relationship between $d(p_1, \ldots, p_k)$ and $d(p'_1, \ldots, p'_k)$ when $elem(\{p_1, \ldots, p_k\}) \subseteq elem(\{p'_1, \ldots, p'_k\})$.
The \textit{separation axiom} $\mathcal{S}_H$ says that, for some nonnegative integer $\gamma\leq k-1$,
\begin{align*}
    d(p_1, \ldots, p_k) &\leq d(p'_1, \ldots, p'_k) \quad \text{if $elem(\{p_1, \ldots, p_k\}) \subset elem(\{p'_1, \ldots, p'_k\})$},\\
    d(p_1, \ldots, p_k) &\leq \gamma \cdot d(p'_1, \ldots, p'_k)\quad \text{if $elem(\{p_1, \ldots, p_k\}) = elem(\{p'_1, \ldots, p'_k\})$}.
\end{align*}

The $H$-metric is a $k$-point metric that satisfies all the above axioms.

\begin{definition}[Melnyk et al.~\cite{melnyk_online_2021}] \label{def:hmetric}
  A $k$-point metric $d_H: \chi^k \rightarrow [0, \infty)$ is an \emph{$H$-metric with parameter $\gamma\leq k-1$} if it satisfies $\Pi$, $O_D$, $\Delta_H$ and $\mathcal{S}_H$ with parameter $\gamma$.
\end{definition}

%

We remark that there are weaker conditions than $\Delta_H$ and $\mathcal{S}_H$, generalizing the triangle inequality, which yields other classes of $k$-point metrics such as the $n$-metric~\cite{assafPartialNmetricSpaces2015} and the $K$-metric~\cite{khanPossibitityNtopologicalSpaces2012}.
See \cite{melnyk_online_2021} for the formal definition.
Melnyk et al.~\cite{melnyk_online_2021}, however, showed that the $k$-MPMD cannot be solved for such more general metrics.
Specifically, they proved that there exists no randomized algorithm for the $k$-MPMD $(k \geq 5)$ problem  on $n$-metric or $K$-metric agaist an oblivious adversary that has a competitive ratio
  which is bounded by a function of the number of points $n$.

\subsection{Properties of $H$-metric}\label{sec:HmetricProperty}

In this section, we discuss approximating $H$-metric by a standard metric, and present specific examples of $H$-metric.

Melnyk et al. proved that $H$-metric can be approximated by the sum of distances between all pairs~\cite[Theorem 6]{melnyk_online_2021}.
We refine their results as in the theorem below, which will be used in the next section.

\begin{theorem} \label{thm:approx_hmetric}
  Let $d_H$ be an $H$-metric on $\chi$ with parameter $\gamma$. 
  Define a metric $d: \chi^2 \rightarrow [0, \infty)$ as 
  \[
  d(p_1, p_2) := d_H(p_1, p_2, \ldots, p_2) + d_H(p_2, p_1, \ldots p_1)
  \]
  for any $p_1, p_2\in \chi$.
  Then it holds that
  \begin{equation}
    \frac{1}{\gamma k^2} \cdot \sum_{i = 1}^{k - 1}\sum_{j=i+1}^{k} d(p_i, p_j) \leq d_H(p_1, \ldots, p_k) \leq \sum_{i=1}^k d(v, p_i), \label{eq:app_hmetric}
  \end{equation}
  for all $v \in \{p_1, \ldots, p_k\}$.
\end{theorem}

\begin{proof}
  We first show that $d$ is a metric.
  The symmetry axiom holds by the definition of $d$ and the positive definiteness axiom comes from $O_D$ of $d_H$.
  The triangle inequality for $p_1, p_2, p_3$ holds by $\Delta_H$ as follows:
  \begin{align*}
    d(p_1, p_3) &= d_H(p_1, p_3, \ldots, p_3) + d_H(p_3, p_1, \ldots, p_1) \\
    &\leq d_H(p_1, p_2, \ldots, p_2) + d_H(p_2, p_3, \ldots, p_3) + d_H(p_3, p_2, \ldots, p_2) + d_H(p_2, p_1, \ldots, p_1) \\
    &= d(p_1, p_2) + d(p_2, p_3).
  \end{align*}
  \noindent
  Therefore, the function $d$ is a metric.
  
  We next prove \eqref{eq:app_hmetric}.
  By symmetry, we may assume that $v = p_1$.
  Then, by applying $\Delta_H$ repeatedly, it holds that
  \begin{align*}
    d_H(p_1, \ldots, p_k) &\leq d_H(p_1, p_2, \ldots, p_{k - 1}, p_1) + d_H(p_1, p_1, \ldots, p_1, p_k) \\
    &\leq d_H(p_1, p_2, \ldots, p_{k - 2}, p_1, p_1) + d_H(p_1, \ldots, p_1, p_{k - 1}, p_1) + d_H(p_1, \ldots, p_1, p_k) \\
    &\leq \dots \leq \sum_{i=2}^{k} d_H(p_1, \ldots, p_1, p_i) \\
    &\leq \sum_{i=2}^{k} (d_H(p_1, \ldots, p_1, p_i) + d_H(p_i, \ldots, p_i, p_1)) \\
    &\leq \sum_{i=2}^{k} d(p_1, p_i) = \sum_{i=1}^{k} d(p_1, p_i). 
  \end{align*}
  Thus the right inequality of \eqref{eq:app_hmetric} holds.

  For the left inequality, since 
  \[
    d_H(p_i, \ldots, p_i, p_j) \leq \gamma \cdot d_H(p_1, p_2, \ldots, p_k)
  \]
  for any $i, j=1,2,\dots, k$ by $\mathcal{S}_H$, we have
  \begin{align*}
    \sum_{i=1}^{k-1} \sum_{j=i+1}^{k} d(p_i, p_j) &= \sum_{i=1}^{k-1} \sum_{j=i+1}^{k} (d_H(p_i, \ldots, p_i, p_j) + d_H(p_j, \ldots, p_j, p_i)) \\
    &\leq \frac{k(k-1)}{2} \cdot 2 \cdot \gamma \cdot d_H(p_1, \ldots, p_k) \leq \gamma k^2 \cdot d_H(p_1, \ldots, p_k).
  \end{align*}
  Thus we obtain the desired lower bound with $c = \frac{1}{\gamma k^2}$.
  \qed
\end{proof}

We conclude this section with providing specific examples of $H$-metric.
We note that the examples below satisfy that $\gamma=1$, and thus the approximation factor in Theorem~\ref{thm:approx_hmetric} becomes small.

  Let $d:\chi^2\to [0, \infty)$ be a distance function.
  We define a $k$-point metric $d_{\max}$ by $d_{\max} (p_1, \dots, p_k)=\max_{i,j\in \{1,\dots, k\}}d(p_i, p_j)$.
  Then it turns out to be an $H$-metric.
  
\begin{proposition}\label{prop:dmax}
  Let $d:\chi^2\to [0, \infty)$ be a distance function.
  Then the $k$-point metric $d_{\max}$ is an $H$-metric with $\gamma=1$.
\end{proposition}

\begin{proof}
  To show that $d_{\max}$ is an $H$-metric, we show $\Pi$, $O_D$, $\mathcal{S}_H$, and $\Delta_H$.
  By definition, $d_{\max}$ clearly satisfies $\Pi$ and $O_D$.
  
  We consider $\mathcal{S}_H$.
  Let $p_1, \dots, p_k, p'_1, \dots, p'_k\in \chi$.
  Then, if $elem(\{p_1, \ldots, p_k\}) \subset elem(\{p'_1, \ldots, p'_k\})$, then $d_{\max}(p_1, \ldots, p_k) \leq d_{\max}(p'_1, \ldots, p'_k)$, and 
  if $elem(\{p_1, \ldots, p_k\}) = elem(\{p'_1, \ldots, p'_k\})$, then $d_{\max}(p_1, \ldots, p_k) = d_{\max}(p'_1, \ldots, p'_k)$ holds.
  Thus $\mathcal{S}_H$ holds with parameter $\gamma =1$.

  It remains to show $\Delta_H$.
  Let $p_1, \dots, p_k\in \chi$ and $i\in\{1,2,\dots, k\}$.
  Suppose that $p, q\in \{p_1, \dots, p_k\}$ satisfy $d(p, q)=d_{\max}(p_1, \dots, p_k)$.
  
  If $p, q \in \{p_1, \ldots, p_i\}$, then we have 
  \begin{align*}
  d_{\max}(p_1, \ldots, p_k) = d(p, q) 
  &\leq d_{\max}(p_1, \ldots, p_i, a, \ldots, a)\\
  & \leq d_{\max}(p_1, \ldots, p_i, a, \ldots, a) + d_{\max}(a, \dots, a, p_{i+1}, \ldots, p_k).
  \end{align*}
  Thus the inequality $\Delta_H$ holds.
  The argument is symmetrical if $p, q \in \{p_{i + 1}, \ldots, p_k\}$.
  
  Suppose that $p \in \{p_1, \ldots, p_i\}$ and $q \in \{p_{i + 1}, \ldots, p_k\}$.
  By definition, 
  \[
      d(p, a) \leq d_{\max}(p_1, \ldots, p_i, a, \ldots, a) \ \  \text{and\ \ }
      d(q, a) \leq d_{\max}(a, \ldots, a, p_{i + 1}, \ldots, p_k).
  \]
  It follows from the triangle inequality of $d$ that
  \begin{align*}
    d_{\max}(p_1, \ldots, p_k) &= d(p, q) \leq d(p, a) + d(a, q) \\
        &\leq d_{\max}(p_1, \ldots, p_i, a, \ldots, a) + d_{\max}(a, \ldots, a, p_{i + 1}, \ldots, p_k).
  \end{align*}
  Thus the axiom $\Delta_H$ is satisfied.
  \qed
\end{proof}

For real numbers $p_1, \dots, p_k\in \mathbb{R}$, we define the \textit{diameter on a line} as $\dD(p_1, \ldots, p_k) = \max_{i, j \in \{1,\dots, k\}} |p_i-p_j|$.
By Proposition~\ref{prop:dmax}, $\dD$ is an $H$-metric.




For a distance function $d: \chi^2 \to [0, \infty)$,
we define another $H$-metric $d_{\mathrm{HC}}$ by
\[
  d_{\mathrm{HC}} (p_1, \dots, p_k) = \min \left\{ \sum_{e \in C} d(e) \mid C \subseteq \binom{\chi}{2}, C \text{ forms a Hamiltonian circuit in } \{ p_1, \dots, p_k \} \right\}
\]
where $\binom{\chi}{2}=\{(p, q)\mid p, q\in \chi, p\neq q\}$.
This means that $d_{\mathrm{HC}} (p_1, \dots, p_k)$ equals to the minimum cost of a Hamiltonian circuit contained in $\{p_1, \dots, p_k\}$ with respect to cost $d$.

\begin{proposition} \label{prop:dHC}
  Let $d: \chi^2 \to [0, \infty)$ be a distance function. Then the k-point
  metric $d_{\mathrm{HC}}$ is an $H$-metric with parameter $\gamma=1$.
\end{proposition}

\begin{proof}
  To show that $d_{\mathrm{HC}}$ is an $H$-metric, we show $\Pi$, $O_D$, $\mathcal{S}_H$, and $\Delta_H$.
  By definition, $d_{\mathrm{HC}}$ satisfy $\Pi$, $O_D$, and $\mathcal{S}_H$ with parameter $\gamma =1$.

  It remains to show that $\Delta_H$ is satisfied.
  Let $a, p_1, \dots, p_k\in \chi$ and $i\in\{1,2,\dots, k\}$.
  
  Let $C_1$ and $C_2$ be minimum cost Hamiltonian circuits for $\{p_1, \dots, p_i, a, \dots, a\}$ and $\{a, \dots, a, p_{i + 1}, \dots, p_k\}$, respectively.
  We here represent $C_1$ and $C_2$ as permutations of requests,
  that is, we denote $C_1 = (a, p_{\pi_1(1)}, \dots, p_{\pi_1(i)}, a)$ and
  $C_2 = (a, p_{\pi_2(i+1)}, \dots, p_{\pi_2(k)}, a)$
  where $\pi_1, \pi_2$ are permutations of $\{1, \dots, i\}$ and $\{i+1, \dots, k\}$, respectively.
  For a circuit $C$, the cost is referred to as $d(C)$.


  Define $C_3$ and $C_4$ as 
  \begin{align*}
  C_3 := (a, p_{\pi_1(1)}, \dots, p_{\pi_1(i)}, a, p_{\pi_2(i+1)}, \dots, p_{\pi_2(k)}, a)\\ 
  C_4 := (p_{\pi_1(1)}, \dots, p_{\pi_1(i)}, p_{\pi_2(i + 1)}, \dots, p_{\pi_2(k)}, p_{\pi_1(1)}).
  \end{align*}
  Then, $d(C_3) = d(C_1) + d(C_2)$ and $d(C_4) \leq d(C_3)$ hold since $d$ satisfies the triangle inequality.
  Since $C_4$ is a Hamiltonian circuit of $\{p_1, \dots, p_k\}$, $d_{\mathrm{HC}}(p_1, \ldots, p_k)\leq d(C_4)\leq d(C_3)$ hold.
  Thus, $d_{\mathrm{HC}}$ satisfies $\Delta_H$.
  \qed
\end{proof}

%

\section{$k$-MPMD on $H$-metric space}\label{sec:algorithm}

This section proposes a primal-dual algorithm for $k$-MPMD on $H$-metric space.
Let $(\chi, d_H)$ be an $H$-metric space with parameter $\gamma$.

\subsection{Linear programming relaxation}

This subsection introduces a linear programming relaxation for computing the offline optimal value $\mathcal{OPT}(\sigma)$ for a given instance $\sigma$.

We first give some notation.
Let $\mathcal{E}=\{F\subseteq V\mid |F|=k\}$.
For any subset $S \subseteq V$, we denote $\sur(S) = |S| \bmod k$, which is the number of remaining requests when we make a $k$-way matching of size $\lfloor|S|/k\rfloor$ among $S$.
We denote $\Delta(S)=\{F\in \mathcal{E}\mid F\cap S \neq\emptyset, F\setminus S\neq \emptyset\}$, which is the family of $k$ request sets that intersect both $S$ and $V\setminus S$.

Preparing a variable $x_F$ for any subset $F\in \mathcal{E}$, 
we define a linear programming problem:
\begin{mpproblem}{$\mathcal{P}$}
  \begin{alignat}{3}
    & \text{min.} & \quad \sum_{F \in \mathcal{E}} & \optcost(F) \cdot x_{F}  & & \notag \\
    & \text{s.t.} & \quad \sum_{F \in \Delta(S)} & x_{F} \geq \left\lceil \frac{\sur(S)}{k} \right\rceil, & & \quad \forall S \subseteq V \label{P_constraint_1} \\
    &             & \quad & x_{F} \geq 0, & & \quad \forall F \in \mathcal{E} \notag 
  \end{alignat}
\end{mpproblem}
where, for any $F=(v_1, \dots, v_k)\in \mathcal{E}$, we define
\[
\optcost(F) :=
d_H(\pos(v_1), \ldots, \pos(v_k)) + \sum_{i = 1}^{k} \left( \max_{j} \atime(v_j) - \atime(v_i) \right).
\]
Notice that $\optcost(F)$ is the cost of choosing $F$ at the moment when all the requests in $F$ have arrived.

Let $\mathcal{M}$ be a perfect $k$-way matching with optimal cost $\mathcal{OPT}(\sigma)$.
Define a $0$-$1$ vector $(x_F)_{F\in \mathcal{E}}$ such that $x_F=1$ if and only if $F\in \mathcal{M}$.
Then the vector satisfies the constraint \eqref{P_constraint_1}.
Moreover, the cost incurred by $F\in \mathcal{M}$ is equal to $\optcost(F)$.
This is because the optimal algorithm that returns $\mathcal{M}$ chooses $F$ at the moment when all the requests in $F$ have arrived.
Thus the objective value for the vector $(x_F)_{F\in \mathcal{E}}$ is equal to $\mathcal{OPT}(\sigma)$, and hence the optimal value of $(\mathcal{P})$ gives a lower bound of $\mathcal{OPT}(\sigma)$.

We further relax the above LP $(\mathcal{P})$ by replacing $x_F$'s with variables for all pairs of requests.
Let $E=\{(u, v)\mid u, v\in V, u\neq v\}$, and we prepare a variable $x_e$ for any $e\in E$.
We often call an element in $E$ an \textit{edge}.

We denote by $\delta(S)$ the set of pairs between $S$ and $V\setminus S$.
Define the following linear programming problem:
\begin{mpproblem}{$\mathcal{P'}$}
  \begin{alignat}{3}
    & \text{min.} & \quad \sum_{e \in E} & \frac{1}{\gamma k^2} \cdot \optcost(e) \cdot x_e & & \notag \\
    & \text{s.t.} & \quad \sum_{e \in \delta(S)} & x_e \geq \sur(S) \cdot (k - \sur(S)), & & \quad \forall S \subseteq V \label{P2_constraint_1}\\
    &             & \quad & x_e \geq 0, & & \quad \forall e \in E \notag
  \end{alignat}
\end{mpproblem}
where, for any $e=(v_1, v_2)\in E$ with $p_1 = \pos (v_1)$ and $p_2 = \pos (v_2)$, we define
\begin{align*}
d(p_1, p_2) &:= d_H(p_1, p_2, \ldots, p_2) + d_H(p_2, p_1, \ldots, p_1),\text{ and }\\
\optcost(e) &:= d(p_1, p_2) + |\atime(v_1) - \atime(v_2)|.
\end{align*}

The following lemma follows from Theorem~\ref{thm:approx_hmetric}.

\begin{lemma} \label{lmm:optcost_hmetric}
  It holds that, for any $F = (v_1, \ldots, v_k)\in \mathcal{E}$,
  \begin{equation}\label{eq:optcost_hmetric}
    \frac{1}{\gamma k^2} \cdot \sum_{i = 1}^{k - 1} \sum_{j = i + 1}^{k} \optcost(v_i, v_j) \leq \optcost(F) \leq \sum_{i=1}^{k} \optcost(v, v_i),
  \end{equation}
  where $v = \argmax_{u \in F} \atime(u)$.
\end{lemma}

\begin{proof}
  We first observe that, by property $\Pi$, we may assume that $\atime(v_1) \leq \atime(v_2) \leq \dots \leq \atime(v_k)$ and $v=v_k$.
  
  By Theorem~\ref{thm:approx_hmetric}, it holds that
  \[
    \frac{1}{\gamma k^2} \cdot \sum_{i = 1}^{k - 1} \sum_{j = i + 1}^{k} d(v_i, v_j) \leq d_H(F) \leq \sum_{i=1}^{k} d(v, v_i).
  \]
  This implies that the right inequality holds since the second term of $\optcost(F)$ is the same as the one of $\sum_{i=1}^k \optcost(v, v_i)$ by definition.

  On the left inequality, we observe that
  \begin{equation*}
    \sum_{i = 1}^{k - 1} \sum_{j = i + 1}^{k} |\atime(v_i) - \atime(v_j)| = \sum_{i = 1}^{k - 1} i \cdot (k - i) \cdot |\atime(v_{i + 1}) - \atime(v_i)|. \label{eq:time_all_pair}
  \end{equation*}
   On the other hand, we have
   \begin{equation*}
   \sum_{i = 1}^{k} \left( \atime(v) - \atime(v_i) \right) = \sum_{i = 1}^{k - 1} i \cdot |\atime(v_{i + 1}) - \atime(v_i)|. \label{eq:time_opt_cost}
   \end{equation*}
  Hence, since $\frac{i \cdot (k - i)}{k} \leq i$ as $1 \leq k - i \leq k$, it holds that
  \begin{align*}
  \frac{1}{k} \sum_{i=1}^{k-1} \sum_{j = i + 1}^{k} |\atime(v_i) - \atime(v_j)| 
  &= \sum_{i = 1}^{k - 1} |\atime(v) - \atime(v_i)|.
  \end{align*}
  Thus we obtain the desired inequality.
  \qed
\end{proof}

For any perfect $k$-way matching $\mathcal{M}$, define an edge subset $M$ such that $e\in M$ if and only if the pair $e$ is contained in some set $F$ of $\mathcal{M}$.
Thus we represent each set in $\mathcal{M}$ with a complete graph of $k$ vertices.
We will show below that the characteristic vector for $M$ is feasible to $(\mathcal{P}')$.
Here, for a subset $X\subseteq E$, the characteristic vector $\mathbbm{1}_X\in\{0,1\}^E$ is defined to be 
\begin{equation*}
\mathbbm{1}_X(x) = 
\begin{cases}
  1 & (x \in X) \\
  0 & (x \notin X) 
\end{cases}.
\end{equation*}
Moreover, this implies that the optimal value of ($\mathcal{P'}$), denoted by $\mathcal{P'}(\sigma)$, is a lower bound of $\mathcal{OPT}(\sigma)$ for any instance $\sigma$.

\begin{lemma} \label{lmm:p_lower_opt}
  Let $\mathcal{M}$ be a perfect $k$-way matching.
  Define an edge subset $M = \{ (u, v) \in E \mid \exists F\in \mathcal{M}~\mathrm{s.t.}~u, v \in F \}$.
  Then $x=\mathbbm{1}_M$ is a feasible solution to $\mathcal{P'}$.
  Furthermore, $\mathcal{P'}(\sigma) \leq \mathcal{OPT}(\sigma)$ holds.
\end{lemma}

\begin{proof}
  We show that $x$ satisfies the constraints of $(\mathcal{P}')$.
  By the definition of $x$, we have $x_e \geq 0$ for any $e\in E$.
  We will show that \eqref{P2_constraint_1} is satisfied by proving $\sum_{e\in \delta(S)}x_e=|\delta(S) \cap M|\geq \sur(S) \cdot (k - \sur(S))$ for any $S\subseteq E$.

  We denote $\mathcal{M} = \{M_1, M_2, \dots, M_p\}$, and define $a_i = |S\cap M_i|$ for $i=1,\dots, p$.
  We note that $0 \leq a_i \leq k$ for $i=1,\dots, p$ and that $|\delta(S) \cap M|=\sum_{i = 1}^{\ell} a_i \cdot (k - a_i)$.
  Consider the worst case by minimizing $\sum_{i = 1}^{\ell} a_i \cdot (k - a_i)$ subject to $\sum_{i=1}^p a_i=|S|$ and $0\leq a_i\leq k$ for any $i=1,\dots, p$.
  We observe that, by letting $\ell = \left\lceil \frac{|S| + 1}{k} \right\rceil$, it is minimized when $a_i = k$ for $i=0,\dots, \ell-1$, $a_\ell = \sur(S) \equiv |S| \pmod k$, and $a_i=0$ for $i\geq \ell+1$.
  Hence $|\delta(S) \cap M|$ is at least $\sur(S) \cdot (k - \sur(S))$, and thus \eqref{P2_constraint_1} is satisfied.
  
  We next show $\mathcal{P'}(\sigma) \leq \mathcal{OPT}(\sigma)$ for a given instance $\sigma$.
  Let $\mathcal{M}^\ast$ be a perfect $k$-way matching with optimal cost $\mathcal{OPT}(\sigma)$.
  We define $M^\ast=\{(u, v) \in E\mid \exists F\in \mathcal{M}^\ast~\mathrm{s.t.}~u, v\in F\}$, and $x=\mathbbm{1}_{M^\ast}$.
  Then, by Lemma~\ref{lmm:optcost_hmetric}, we obtain
  \begin{align*}
    \mathcal{P'}(\sigma) &\leq \sum_{e \in E} \frac{1}{\gamma k^2} \cdot \optcost(e) \cdot x_e  = \sum_{F \in \mathcal{M}^\ast} \frac{1}{\gamma k^2} \cdot \sum_{e = (u, v) : u, v \in F} \optcost(e) \\
    &\leq \sum_{F \in \mathcal{M}^\ast} \optcost(F) = \mathcal{OPT}(\sigma).
  \end{align*}
  Thus the lemma holds.
  \qed
\end{proof}

The dual linear programming problem of $(\mathcal{P'})$ is 
\begin{mpproblem}{$\mathcal{D'}$}
  \begin{alignat}{3}
    & \text{max.} & \quad \sum_{S \subseteq V} & \sur(S) \cdot (k - \sur(S)) \cdot y_S & & \notag \\
    & \text{s.t.} & \quad \sum_{S : e \in \delta(S)} & y_S \leq \frac{1}{\gamma k^2} \cdot \optcost(e), & & \quad \forall e \in E \label{constraint_D}\\
    &             & \quad & y_S \geq 0, & & \quad \forall S \subseteq V \notag
  \end{alignat}
\end{mpproblem}
The weak duality of LP implies that $\mathcal{D}'(\sigma) \leq \mathcal{P}'(\sigma)$, where $\mathcal{D}'(\sigma)$ is the dual optimal value.

\subsection{Greedy Dual for \textit{k}-MPMD (GD-\textit{k})}

We present our proposed algorithm, called \textit{Greedy Dual for \textit{k}-MPMD}(GD-\textit{k}).
The proposed algorithm extends the one by Bienkowski et al.~\cite{bienkowski_primal-dual_2018} for $2$-MPMD using the LP $(\mathcal{P}')$.

In the algorithm GD-\textit{k}, we maintain a family of subsets of requests, called \textit{active sets}.
At any time, any request $v$ arrived so far belongs to exactly one active set, denoted by $A(v)$.
We also maintain a $k$-way matching $\mathcal{M}$.
A request not in $\bigcup_{F\in \mathcal{M}}F$ is called \textit{free}, and, for a subset $S\subseteq V$ of requests, $\free(S)$ is the set of free requests in $S$.

When request $v$ arrives, we initialize $A(v) = \{v\}$ and $y_S=0$ for any subset $S\subseteq V$ such that $v\in S$.
At any time, for an active set $S$ such that $\free(S)$ is nonempty, we increase $y_S$ with rate $r$, where $r$ is set to be $1/(\gamma k^2)$.
Then, at some point, there exists an edge $e=(u, v)\in E$ such that $\sum_{S : e \in \delta(S)}y_S = \frac{1}{\gamma k^2} \cdot \optcost(e)$, which we call a \textit{tight} edge.
When it happens, we merge the active sets $A(u)$ and $A(v)$ to a large subset $S=A(u)\cup A(v)$, that is, we update $A(w)=S$ for all $w\in S$.
We also mark the tight edge $e$.
If $|\free(S)| \geq k$, we partition $\free (S)$ arbitrarily into subsets of size $k$ with $\sur (S)$ free requests, and add these size-$k$ subsets to $\mathcal{M}$.

The pseudo-code of the algorithm is given as in Algorithm~\ref{alg:GD-k}.

Let $T$ be the time when all requests are matched in the algorithm.
For any subset $S$, we denote the value of $y_S$ at time $\tau$ in the algorithm by $y_S(\tau)$.

\begin{algorithm}[htb]
  \caption{Greedy Dual for \textit{k}-MPMD}
  \label{alg:GD-k}
  \begin{algorithmic}[1]
    \Procedure{GD-\textit{k}}{$\sigma$}
      \State $\mathcal{M} \gets \emptyset$
      \ForAll{moments $t$}
        \If{a request $v$ arrives}
          \State $A(v) \gets \{v\}$
          \ForAll{subsets $S \ni v$}
            \State $y_S \gets 0$
          \EndFor
          \State \textbf{modify} constraints of $(\mathcal{D}')$.
        \EndIf
        \If{there exists $e = (u, v)\in E$ such that $\sum_{S : e \in \delta(S)} y_S = \frac{1}{\gamma k^2} \cdot \optcost(e)$ and $A(u) \neq A(v)$}
          \State $S \gets A(u) \sqcup A(v)$
          \ForAll{$v \in S$}
            \State $A(v) \gets S$
          \EndFor
          \State \textbf{mark} $e$
          \While{$|\free(S)| \geq k$}
            \State choose arbitrarily a set $F$ of $k$ requests from $S$
            \State $\mathcal{M} \gets \mathcal{M}\cup \{F\}$
          \EndWhile
        \EndIf
        \ForAll{sets $S$ which are active and $\free(S) \neq \emptyset$}
          \State increase continuously $y_S$ at the rate of $r$ per unit time
        \EndFor
      \EndFor
    \EndProcedure
  \end{algorithmic}
\end{algorithm}


We show that $y_S$'s maintained in Algorithm~\ref{alg:GD-k} are always dual feasible.

\begin{lemma}\label{lmm:sum_yS_lower_rt}
  For any request $v$, it holds that 
  \begin{equation}
  \sum_{S: v \in S} y_S(\tau) \leq r \cdot (\tau - \atime(v)) \label{eq:yS_lower_rt}
  \end{equation}
  at any time $\tau \geq \atime(v)$.
  This holds with equality while $v$ is not matched.
\end{lemma}

\begin{proof}
When $v$ arrives, Algorithm~\ref{alg:GD-k} initializes $y_S = 0$ for any $S$ with $v \in S$.
Thus the both sides of~\eqref{eq:yS_lower_rt} are $0$ at time $\tau = \atime(v)$.

Suppose that $\tau > \atime(v)$.
Then $v$ belongs to exactly one active set $A(v)$.
If $v$ has not been matched so far, then we increase $y_{A(v)}(\tau)$ with rate $r$, and $y_S(\tau)$ remains unchanged for any other subsets $S$ such that $v\in S$.
Hence the left-hand side of~\eqref{eq:yS_lower_rt} is increased with rate $r$, implying that $\sum_{S: v \in S} y_S(\tau) = r \cdot (\tau - \atime(v))$.

After $v$ has been matched, an active set $S := A(v)$ is whether $\free(S) = \emptyset$ or not.
$y_S(\tau)$ increases at most with rate $r$ in both cases.
Hence we have $\sum_{S: v \in S} y_S(\tau) \leq r \cdot (\tau - \atime(v))$.
  \qed
\end{proof}

\begin{lemma} \label{lmm:feasible}
  Let $r = \frac{1}{\gamma k^2}$. 
  Then, at any time $\tau$, $y_S(\tau)$ maintained in Algorithm~\ref{alg:GD-k} is a feasible solution to $(\mathcal{D'})$.
\end{lemma}

\begin{proof}
  In the algorithm, $y_S$ is non-decreasing with an initial value $0$. Thus $y_S\geq 0$ for any subset $S$. 
  The rest of the proof is devoted to showing that the constraint~\eqref{constraint_D} is satisfied by induction on time.
  Suppose that the constraint~\eqref{constraint_D} is satisfied just before time $\tau$.

  Suppose that a new request $v$ arrives at time $\tau$, i.e., $\atime(v)=\tau$.
  At this point, $(\mathcal{D'})$ makes new variables $y_{\{v\} \cup S}$ for any subset $S$ of already arrived requests, all of which are set to be $0$.

  Let $e=(u, w)$ be a pair of requests such that $\atime(u) \leq \tau$ and $\atime(w) \leq \tau$.
  If $u$ or $w$ are not equal to $v$, then the constraint \eqref{constraint_D} remains satisfied, since we have
  \begin{align*}
    \sum_{S : e \in \delta(S)} y_S(\tau)
    &= \sum_{S : e \in \delta(S), v \notin S} y_S (\tau)
    +\sum_{S : e \in \delta(S), v \in S} y_S(\tau)\\
    &= \sum_{S : e \in \delta(S), v \notin S} y_S(\tau)
    \leq \frac{1}{\gamma k^2} \cdot \optcost(e).
  \end{align*}
  Suppose that $e=(u, v)$.
  Then, since $r = \frac{1}{\gamma k^2}$, it holds by Lemma~\ref{lmm:sum_yS_lower_rt} that
  \begin{align*}
    \sum_{S : e \in \delta(S)} y_S(\tau)
    = \sum_{S : u \in S, v \notin S} y_S(\tau)
    &\leq \frac{1}{\gamma k^2} \cdot (\atime(v) - \atime(u)) \\
    & \leq \frac{1}{\gamma k^2} \cdot (d(\pos(u), \pos(v)) + |\atime(u) - \atime(v)|)\\
    & = \frac{1}{\gamma k^2} \cdot \optcost(e).
  \end{align*}
  Thus the constraint~\eqref{constraint_D} holds when a new request $v$ arrives at time $\tau$.
  
  Suppose that no request arrives at time $\tau$.
  Let $e=(u, v)$ be a pair of requests such that $\atime(u) \leq \tau$ and $\atime(v) \leq \tau$.
  We observe that, if the pair $e$ is tight, $y_S(\tau)$ does not increase for any subset $S$ with $e \in \delta(S)$, since no active set $S$ satisfies $e \in \delta(S)$ by Algorithm~\ref{alg:GD-k}.
  Thus, after the pair $e=(u, v)$ becomes tight, $\sum_{S : e \in \delta(S)} y_S(\tau)$ does not increase, implying that the constraint~\eqref{constraint_D} keeps satisfied.
  
  Therefore, $y_S$'s are dual feasible at any time.
  \qed
\end{proof}

\subsection{Competitive Ratio of GD-$k$}

To bound the competitive ratio of GD-\textit{k}, we evaluate the distance cost and the waiting cost separately.
We will show that each cost is upper-bounded by the dual optimal value of $\mathcal{D'}(\sigma)$.

\subsubsection{Waiting Cost}

We can upper-bound the waiting cost of the output as follows.
\begin{lemma} \label{lmm:time_cost}
  Let $\mathcal{M}=\{M_1, \dots, M_p\}$ be a perfect $k$-way matching returned by Algorithm~\ref{alg:GD-k}, and let $\tau_\ell$ be the time when we match $M_\ell$.
  Then it holds that
  \[
     \sum_{\ell=1}^{p}\sum_{i=1}^k (\tau_\ell - \atime(v_{\ell, i})) = \frac{1}{r} \cdot \sum_{S \subseteq V} \sur(S) \cdot y_S(T) \leq \frac{1}{r} \cdot \mathcal{D'}(\sigma),
  \]
    where we denote $M_\ell = \{v_{\ell,1}, \dots, v_{\ell, k}\}$.
\end{lemma}

\begin{proof}
Consider sufficiently small time period $\Delta t$ that no new request arrives and no pair becomes tight.
Let $S$ be an active set that contains free requests.
Then, $S$ has $\sur (S)$ free requests by the algorithm, and they wait during the time period $\Delta t$, and hence the waiting cost incurred by these requests in this period is $\sur (S) \cdot \Delta t$.
This implies that  the total waiting cost is
\[
    \int_0^T \sum_{S\in \mathcal{A}(\tau)} \sur(S)~\mathrm{d\tau},
\]
where $\mathcal{A}(\tau)$ is the family of the active sets that contain free requests at time $\tau$.
Defining $\mathbbm{1}_{\mathcal{A}(\tau)}(S)$ as $1$ if $S\in \mathcal{A}(\tau)$ and $0$ otherwise, we see that it is equal to
  \begin{align*}
    \int_0^T \sum_{S\subseteq V} \sur(S) \cdot \mathbbm{1}_{\mathcal{A}(\tau)}(S)~\mathrm{d\tau} 
    &= \sum_{S\subseteq V} \sur(S) \int_0^T \mathbbm{1}_{\mathcal{A}(\tau)}(S)~\rm{d\tau} \\
    &= \frac{1}{r} \cdot \sum_{S\subseteq V} \sur(S)\cdot y_S(T),
  \end{align*}
  where the last equality follows from the observation that Algorithm~\ref{alg:GD-k} increases $y_S$ by $r \cdot \Delta t$ for $S\in\mathcal{A}(\tau)$ during sufficiently small time period $\Delta t$.
  Since $1\leq k - \sur(S)$ by $\sur(S) < k$, the total waiting cost is upper-bounded by $\frac{1}{r} \sum_S \sur(S) \cdot (k - \sur(S)) \cdot y_S(T) = \frac{1}{r}  \mathcal{D'}(\sigma)$.
  \qed
\end{proof}

\subsubsection{Distance Cost}

We say that a set $S\subseteq V$ is \textit{formerly-active at time} $\tau$ if $S$ is not active at time $\tau$, but has been active before time $\tau$.

\begin{lemma} \label{lmm:mark_spanning_tree}
  Let $S$ be an active or formerly-active set at time $\tau$.
  Then, marked edges both of whose endpoints are contained in $S$ form a spanning tree in $S$.
\end{lemma}

\begin{proof}
  At time $0$, there are no marked edges, and hence the statement holds.
  Suppose that the statement holds before time $\tau$.
  Also, an active or formerly-active set $S$ of size $1$ includes no marked edges, which means that the statement holds. 

  Consider the moment when some pair $e = (u, v)$ becomes tight and $A(u) \neq A(v)$.
  Then we merge two active sets $A(u)$ and $A(v)$ into one active set $S=A(u)\cup A(v)$, and mark the edge $e=(u, v)$.
  By induction, the marked edges in $A(u)$ and $A(v)$ form spanning trees, respectively.
  Since $A(u)$ and $A(v)$ are disjoint, the spanning trees are disjoint.
  Hence the two spanning trees with the edge $e$ form a spanning tree in $S$.
  \qed
\end{proof}

We now evaluate the distance cost.

\begin{lemma} \label{lmm:dist_cost}
  Let $\mathcal{M}=\{M_1, \dots, M_p\}$ be a perfect $k$-way matching returned by Algorithm~\ref{alg:GD-k}.
  Then it holds that
  \[
     \sum_{\ell=1}^{p} d(\pos(v_{\ell,1}), \dots, \pos(v_{\ell, k})) \leq 4\gamma mk\cdot \sum_S \sur(S) \cdot (k - \sur(S)) \cdot y_S(T) \leq 4\gamma mk\mathcal{D'}(\sigma),
  \]
  where we denote $M_\ell = \{v_{\ell,1}, \dots, v_{\ell, k}\}$.
\end{lemma}

\begin{proof}
  Suppose that, at time $\tau$, some pair $e=(u, v)$ becomes tight and that $|\free (S)|\geq k$, where $S=A(u)\cup A(v)$.
  Then we choose a set $Y$ of $k$ requests arbitrarily from $\free (S)$, which is added to $\mathcal{M}$.
  
  We will evaluate the cost of choosing $Y$ at time $\tau$.
  For simplicity, we denote $Y=\{v_1, \dots, v_k\}$, where $\atime (v_1)\geq \atime (v_i)$ for any $i=1,\dots, k$. 
  We also denote $\pos (v_i)=p_i$ for $i=1,\dots, k$.
  
  Let $F$ be the set of marked edges in $S$.
  By Lemma~\ref{lmm:mark_spanning_tree}, $F$ forms a spanning tree in $S$. 
  Since $F$ is a spanning tree, for any pair $v_i, v_j$ in $Y$, there exists a unique path $P_{i, j}$ from $v_i$ to $v_j$ along $F$.
  By the triangle inequality, $d(v_i, v_j)$ is upper-bounded by the total length of $P_{i, j}$.
  By Theorem~\ref{thm:approx_hmetric}, it holds that
  \begin{align*}
    d_H(p_1,\dots, p_k) &\leq \sum_{i=1}^k d(p_1, p_i) = \sum_{i = 2}^{k} d(p_1, p_i) \leq \sum_{i = 2}^{k} \sum_{e \in P_{1, i}} d(e)\\
    &\leq \sum_{i=2}^k \gamma k^2 \sum_{e \in P_{1, i}} \frac{1}{\gamma k^2} \cdot \optcost(e),    
  \end{align*}
  where the last inequality follows since $d(e)\leq \optcost(e)$.
  Since each pair in $P_{1, i}$ is tight, it is equal to 
  \[
    \gamma k^2 \cdot \sum_{i = 2}^{k} \sum_{e \in P_{1, i}} \sum_{S' : e \in \delta(S')} y_{S'}(\tau)
     \leq \gamma k^2 \sum_{i = 2}^{k} \sum_{S'\subseteq V} | P_{1, i} \cap \delta(S') | \cdot y_{S'}(\tau).
  \]

  \begin{claim}\label{clm:1}
  Let $S'$ be an active or formerly-active subset at time $\tau$.
  Then it holds that $|P_{i, j} \cap \delta(S')| \leq 2$ for any $v_i$ and $v_j$.
  \end{claim}
  \begin{proof}
  Suppose to the contrary that $|P_{i, j} \cap \delta(S')| > 2$ for some $v_i$ and $v_j$.
  This means that there exist $u, w \in S', u', w' \notin S'$ such that $P_{i,j}$ forms $P_{i, j} = (v_i, \ldots, u, u', \ldots, w', w, \ldots, v_j)$ and the subpath from $u$ to $w$ is internally disjoint from $S'$.

  Since $S'$ has a spanning tree formed by the marked edges, $S'$ contains a path from $u$ to $w$ with the marked edges.
  This path, together with $P_{i,j}$, forms a cycle with the marked edges.
  Since active sets are disjoint and $\delta (S)$ has no marked edges for any active subset $S$, this contradicts Lemma~\ref{lmm:mark_spanning_tree}.
  \qed
  \end{proof}

  Let $\mathcal{S}$ be the family of active or formerly-active sets at time $\tau$.
  We observe that, if $y_{S'}(\tau) > 0$, then $S'$ is active or formerly-active.
  Hence the above claim implies that
  \[
    d_H(p_1, \ldots, p_k) 
    \leq \gamma k^2 \sum_{i = 2}^{k} \sum_{S' \in \mathcal{S}} | P_{1, i} \cap \delta(S') | \cdot y_{S'}(\tau)
    \leq 2\gamma k^2 \sum_{i = 2}^{k} \sum_{S' \in \mathcal{S}} y_{S'}(T),
   \]
   since $y_{S'}(\tau)\leq y_{S'}(T)$ for any subset $S'$.
  Since $k \leq 2 \cdot \sur(S') \cdot (k - \sur(S'))$ as $1 \leq \sur(S') \leq k - 1$, it holds that
  \begin{align*}
    d_H(p_1, \ldots, p_k) 
	    &\leq    4 \cdot \gamma k \sum_{i=2}^{k} \sum_{S' \in \mathcal{S}} \sur(S') \cdot (k - \sur(S')) \cdot y_{S'}(T) \\
    &\leq 4 \cdot \gamma k \cdot (k - 1) \sum_{S' \in \mathcal{S}} \sur(S') \cdot (k - \sur(S')) \cdot y_{S'}(T) \\
    &\leq 4 \gamma k^2 \cdot \sum_{S' \in \mathcal{S}} \sur(S') \cdot (k - \sur(S')) \cdot y_{S'}(T)
    = 4 \gamma k^2 \mathcal{D'}(\sigma).
  \end{align*}

  Since the total number of requests is $m$, the final $k$-way matching has $\frac{m}{k}$ subsets.
  Therefore, the total distance cost is at most
    \[
    \frac{m}{k} \cdot 4\gamma k^2 \mathcal{D'}(\sigma) = 4\gamma mk \mathcal{D'}(\sigma).
  \]
  Thus the lemma holds.
  \qed
\end{proof}

\subsubsection{Competitive Ratio}

Summarizing the above discussion, we obtain Theorem~\ref{thm:GD-k-comp}.

\begin{theorem} \label{thm:GD-k-comp}
  Let $d_H$ be an $H$-metric with parameter $\gamma$.
  Setting $r=1/(\gamma k^2)$, Greedy Dual for $k$-MPMD achieves a competitive ratio $(4mk + k^2)\gamma$ for $k$-MPMD.
\end{theorem}

\begin{proof}
  Let $\sigma$ be an instance of $k$-MPMD.
  It follows from Lemmas~\ref{lmm:time_cost} and~\ref{lmm:dist_cost} that 
  the cost of the returned perfect $k$-way matching is upper-bounded by $(4mk + k^2)\gamma\cdot\mathcal{D'}(\sigma)$.
  By the weak duality and Lemma~\ref{lmm:feasible}, we observe that $\mathcal{D'}(\sigma)\leq \mathcal{P'}(\sigma) \leq \mathcal{OPT}(\sigma)$.
  Thus the theorem holds.
  \qed
\end{proof}


Finally, we consider applying our algorithm to the problem with specific $H$-metrics such as $d_{\max}$ and $d_{\mathrm{HC}}$ given in Section~\ref{sec:HmetricProperty}.
Since they have parameter $\gamma =1$, it follows from Theorem~\ref{thm:GD-k-comp} that GD-$k$ achieves a competitive ratio $O(mk + k^2)$.
In the case of $d_{\max}$, we can further improve the competitive ratio.

\begin{theorem}\label{thm:diameter}
  For the $k$-MPMD on a metric space $(\chi, d_{\max})$, GD-$k$ achieves a competitive ratio $O(m + k^2)$.
\end{theorem}

\begin{proof}
  In this case, the parameter $\gamma$ is $1$.
  By Lemma~\ref{lmm:time_cost}, we see that the waiting cost can be upper-bounded by $\frac{1}{r}\mathcal{D}'(\sigma) = k^2 \mathcal{D}'(\sigma)$.
  It remains to show that the distance cost is bounded by $m \mathcal{D}'(\sigma)$.
  
  We follow the proof of Lemma~\ref{lmm:dist_cost}, using the same notation.
  Suppose that, at time $\tau$, some pair $e=(u, v)$ becomes tight and that $|\free (S)|\geq k$, where $S=A(u)\cup A(v)$.
  Let $Y\subseteq \free (S)$ be a set of size $k$, which is added to $\mathcal{M}$.
  For simplicity, we denote $Y=\{v_1, \dots, v_k\}$ and $\atime (v_1)\geq \atime (v_i)$ for any $i=1,\dots, k$.
  Let $F$ be the spanning tree formed by the marked edges in $S$.  
  Then, by the definition of the diameter on a line,
  $d(v_i, v_j)$ is upper-bounded by the diameter of $F$.
  Let $P$ be a path of $F$ whose length is equal to the diameter.
  
  Applying a similar argument to Lemma~\ref{lmm:dist_cost} with $\gamma=1$,
  it holds that
  \begin{align*}
    d_{\max}(\pos(v_1), \ldots, \pos(v_k)) &= \max_{i, j \in \{1,2,\dots, k\}} |\pos(v_i) - \pos(v_j)| \\
    &\leq \sum_{e \in P} d(e) \\
    &\leq k^2 \sum_{e \in P} \frac{1}{k^2} \cdot \optcost(e) \\
    &\leq k^2 \sum_{S' \in \mathcal{S}} |P \cap \delta(S')| \cdot y_{S'}(T) \\
    &\leq k \sum_{S' \in \mathcal{S}} 4 \cdot \sur(S') \cdot (k - \sur(S')) \cdot y_{S'}(T) \leq 4k\mathcal{D'}(\sigma).
  \end{align*}
  Since the number of matching is $\frac{m}{k}$, the total distance cost is $\frac{m}{k} \cdot 4k\mathcal{D'}(\sigma)=4m\mathcal{D'}\sigma$.

  Therefore, the total cost of GD-\textit{k} is at most $(4m + k^2) \cdot \mathcal{D'}(\sigma)$, implying that the competitive ratio is $O(m + k^2)$.
  \qed
\end{proof}

\section{Lower bound of GD-$k$ for a diameter on a line}\label{sec:lowerbound}

In this section, we show a lower bound on the competitive ratio for GD-$k$ for the metric $\dD$.
Recall that $\dD(p_1,\dots, p_k)=\max_{i, j\in\{1,\dots, k\}}|p_i-p_j|$ for $p_1,\dots, p_k\in\mathbb{R}$.

We define an instance $\sigma_l =(V, \pos, \atime)$ where $V=\{u_1, u_2, \dots, u_m\}$ as follows.
Suppose that the number $m$ of requests is equal to $m = sk^2$ for some integer $s$.
Let $p_1, \dots, p_k$ be $k$ points in $\mathbb{R}$ such that $d(p_i, p_{i+1}) = 2$ for any $i=1, 2, \dots, k-1$.
 
For $i=1, 2, \dots, sk$ and $j=1,2,\dots, k$, define $\atime (u_{k(i-1)+j}) = t_i$ and $\pos (u_{k(i-1)+j}) =p_j$ for $j=1, 2, \dots, k$, 
where we define $t_1=0$ and $t_i=1 + (2 i - 3)\varepsilon$ for $i\geq 2$.
Thus, at any time $t_i$~($i=1, \dots, sk$), the $k$ requests $u_{k(i-1)+1}, \dots, u_{k(i-1)+k}$ arrive at every point in $p_1, \dots, p_k$, respectively.

Then it holds that $\mathcal{OPT}(\sigma_l)\leq k+k\varepsilon + k^3\varepsilon + mk\varepsilon$, while the output of GD-$k$ has cost at least $m + k + (m - k) \varepsilon$.

\begin{theorem} \label{thm:lower_bound}
  For a metric space $(\mathbb{R}, \dD)$, 
  there exists an instance $\sigma_l$ of $m$ requests such that GD-$k$ admits a competitive ratio $\Omega(\frac{m}{k})$.
\end{theorem}

\begin{proof}
We begin with the following claim.

\begin{claim} \label{lmm:WAIT_cost}
  It holds that $\mathcal{OPT}(\sigma_l)\leq k+k\varepsilon + k^3\varepsilon + mk\varepsilon$.
\end{claim}

\begin{proof}
Consider an algorithm that we repeatedly match $k$ requests at the same points at the moment when there are $k$ free requests.
That is, for $j=1,2,\dots, k$, we match $M_{jh}=\{u_{hk^2+j}, u_{hk^2+ k +j}, \dots, u_{hk^2+k(k - 1)+j}\}$ for $h=0,1,\dots, s-1$.
The algorithm returns $\{M_{jh}\mid j\in\{1,2,\dots, k\}, h\in\{0,1,\dots, s-1\}\}$ as a perfect $k$-way matching.

We calculate the cost of choosing $M_{jh}$.
Since the distance cost is clearly zero, we focus on the waiting cost.

First consider when $h=0$.
Then we observe that $\atime (u_{k+j})-\atime (u_j) = 1+\varepsilon$ and $\atime (u_{ki+j})-\atime (u_{k(i-1)+j}) = 2\varepsilon$ for $i=2, \dots, k-1$.
Hence the cost of choosing $M_{j0}$ is equal to
\begin{align*}
\sum_{i=1}^{k-1} \left(\atime (u_{k(k-1)+j})-\atime (u_{k(i-1)+j})\right)
&= \sum_{i=1}^{k-1} i \left(\atime (u_{ki+j})-\atime (u_{k(i-1)+j})\right)\\
&= 1 +\varepsilon + \sum_{i=2}^{k-1}i \cdot 2\varepsilon\\
& = 1 +\varepsilon+ \left( k(k - 1) - 2\right) \varepsilon.
\end{align*}

Next consider the case when $h=1, 2, \dots, s-1$.
Since $\atime (u_{hk^2+ki+j})-\atime (u_{hk^2+k(i-1)+j}) = 2\varepsilon$ for $i=1,2, \dots, k-1$, we have
\[
\sum_{i=1}^{k-1} \left(\atime (u_{hk^2+k(k-1)+j})-\atime (u_{hk^2+k(i-1)+j})\right)
=\sum_{i = 1}^{k - 1} \left( i \cdot 2 \varepsilon \right) = k(k - 1) \varepsilon.
\]

Therefore, since $s-1=m/k^2-1$, the total cost of the perfect $k$-way matching $\{M_{jh}\mid j\in\{1,2,\dots, k\}, h\in\{0,1,\dots, s-1\}\}$ is 
\begin{align*}
k \cdot \left(1 +\varepsilon+ (k(k - 1) - 2)\varepsilon + \left( \frac{m}{k^2} - 1 \right) k(k - 1) \varepsilon\right)
\leq k\left(1 +\varepsilon+ k^2\varepsilon + m\varepsilon\right).
\end{align*}
  \qed
\end{proof}

We next estimate the cost by GD-$k$.
Let $\text{GD-\textit{k}}(\sigma_l)$ be the cost of the output that GD-$k$ returns.
\begin{claim} \label{lmm:gd_k_cost_lower}
  For $k\geq 2$, it holds that
    $\text{GD-\textit{k}}(\sigma_l) \geq m + k + (m - k) \varepsilon$.
\end{claim}

\begin{proof}
Suppose that we run the algorithm GD-$k$ to $\sigma_l$.
Initial active sets are $A(u_j)=\{u_j\}$ for $j=1,2,\dots, k$.
We gradually increase $y_{\{u_j\}}$ in the algorithm.
Then, at time $1$, each pair $(u_i, u_{i+1})$ becomes tight for any $i=1,2,\dots, k-1$.
This implies that we obtain the active set $M_1 = \{u_1, \dots, u_k\}$, which is added to $\mathcal{M}$ at time $1$.
We now have no active sets.

At time $1+\varepsilon$, new requests $u_{k+1}, \dots, u_{2k}$ arrive.
We gradually increase $y_{\{u_{k+j}\}}$ for $j=1,2,\dots, k$.
Then, at time $1+2\varepsilon$, each pair $e=(u_j, u_{k+j})$ becomes tight for any $j=1,2,\dots, k$, since at time $1+2\varepsilon$,
\[
\sum_{S:e\in \delta (S)}y_S = y_{M_1} + y_{\{u_{k+j}\}} = 1 + \varepsilon = \optcost (e).
\]
We merge $M_1$ and all $\{u_{k+j}\}$'s and add $M_2 = \{u_{k+1}, \dots, u_{2k}\}$ to $\mathcal{M}$.
The algorithm proceeds for $i\geq 2$.
Specifically, at time $1 + (2 i - 3)\varepsilon+\varepsilon$, 
each pair $(u_{k(i-2)+j}, u_{k(i-1)+j})$ becomes tight for any $j=1,2,\dots, k$, and $M_i = \{u_{k(i-1)+j}\mid j=1, 2, \dots, k\}$ is added to $\mathcal{M}$.

Since each $M_i$ has the distance cost $2(k-1)$, the total distance cost is $2 \cdot (k - 1) \cdot \frac{m}{k}$.
The waiting cost for choosing $M_1$ is $k$.
Since each request of $M_i$ for $i\geq 2$ waits for $\varepsilon$ time,
the waiting cost of each $M_i$ for $i\geq 2$ is $k\varepsilon$.
Hence the total waiting cost is $k + (m - k) \varepsilon$.
Therefore, since $\frac{k-1}{k}\geq \frac{1}{2}$ as $k\geq 2$, the total cost is 
\[
2\frac{m(k - 1)}{k} + k + (m - k) \varepsilon\geq m+k + (m-k)\varepsilon.
\]
  \qed
\end{proof}

It follows from the above two claims that the competitive ratio is at least
\[
\frac{m+k + (m-k)\varepsilon}{k(1+\varepsilon + k^2\varepsilon + m\varepsilon)}
\geq
\frac{m+k}{4k}
\]
if $\varepsilon \leq 1/\max\{k^2, m\}$.
Thus the competitive ratio is $\Omega(\frac{m}{k})$.
\qed
\end{proof}



%
%
\bibliographystyle{splncs04}
\bibliography{citation_details}

\begin{thebibliography}{10}
\providecommand{\url}[1]{\texttt{#1}}
\providecommand{\urlprefix}{URL }
\providecommand{\doi}[1]{https://doi.org/#1}

\bibitem{aggarwalOnlineVertexWeightedBipartite2011}
Aggarwal, G., Goel, G., Karande, C., Mehta, A.: Online {{Vertex-Weighted
  Bipartite Matching}} and {{Single-bid Budgeted Allocations}}. In: Proceedings
  of the 2011 {{Annual ACM-SIAM Symposium}} on {{Discrete Algorithms}}, pp.
  1253--1264. SODA '11, {SIAM} (2011)

\bibitem{antoniadis2019left}
Antoniadis, A., Barcelo, N., Nugent, M., Pruhs, K., Scquizzato, M.: A
  $o(n)$-competitive deterministic algorithm for online matching on a line.
  Algorithmica  \textbf{81},  2917--2933 (2019)

\bibitem{ashlagiEdgeWeightedOnlineWindowed2023}
Ashlagi, I., Burq, M., Dutta, C., Jaillet, P., Saberi, A., Sholley, C.:
  Edge-{{Weighted Online Windowed Matching}}. Mathematics of Operations
  Research  \textbf{48}(2),  999--1016 (2023)

\bibitem{assafPartialNmetricSpaces2015}
Assaf, S., Pal, K.: Partial n-metric spaces and fixed point theorems (2015),
  arXiv: 1502.05320

\bibitem{azar_polylogarithmic_2017}
Azar, Y., Chiplunkar, A., Kaplan, H.: Polylogarithmic bounds on the
  competitiveness of min-cost perfect matching with delays. In: Proceedings of
  the {{Twenty-Eighth Annual ACM-SIAM Symposium}} on {{Discrete Algorithms}}
  ({{SODA}} '17). pp. 1051--1061. {SIAM} (2017)

\bibitem{azar_deterministic_2020}
Azar, Y., Jacob~Fanani, A.: Deterministic {{Min-Cost Matching}} with
  {{Delays}}. Theory of Computing Systems  \textbf{64}(4),  572--592 (2020)

\bibitem{azarMinCostMatchingConcave2021}
Azar, Y., Ren, R., Vainstein, D.: The {{Min-Cost Matching}} with {{Concave
  Delays Problem}}. In: Proceedings of the 2021 {{ACM-SIAM Symposium}} on
  {{Discrete Algorithms}}, pp. 301--320. SODA '21, {SIAM} (2021)

\bibitem{bansalLog2kCompetitiveAlgorithm2007a}
Bansal, N., Buchbinder, N., Gupta, A., Naor, J.S.: An ${O}(\log^2
  k)$-{{Competitive Algorithm}} for {{Metric Bipartite Matching}}. In:
  Algorithms \textendash{} {{ESA}} 2007. LNCS, vol.~4698, pp. 522--533.
  {Springer} (2007)

\bibitem{bienkowski_primal-dual_2018}
Bienkowski, M., Kraska, A., Liu, H.H., Schmidt, P.: A {{Primal-Dual Online
  Deterministic Algorithm}} for {{Matching}} with {{Delays}}. In: 16th
  {Workshop} on {Approximation} and {Online} {Algorithms} ({WAOA} 2018). LNCS,
  vol. 11312, pp. 51--68. {Springer} (2018)

\bibitem{deryckere2023online}
Deryckere, L., Umboh, S.W.: Online matching with set and concave delays (2023),
  arXiv: 2211.02394

\bibitem{emek_online_2016}
Emek, Y., Kutten, S., Wattenhofer, R.: Online matching: Haste makes waste! In:
  Proceedings of the Forty-Eighth Annual {{ACM}} Symposium on {{Theory}} of
  {{Computing}}. pp. 333--344. {{STOC}} '16, {ACM} (2016)

\bibitem{fuchsOnlineMatchingLine2005}
Fuchs, B., Hochst{\"a}ttler, W., Kern, W.: Online matching on a line.
  Theoretical Computer Science  \textbf{332}(1),  251--264 (2005)

\bibitem{goelOnlineBudgetedMatching2008a}
Goel, G., Mehta, A.: Online budgeted matching in random input models with
  applications to {{Adwords}}. In: Proceedings of the Nineteenth Annual
  {{ACM-SIAM}} Symposium on {{Discrete}} Algorithms. pp. 982--991. {{SODA}}
  '08, {SIAM} (2008)

\bibitem{guptaOnlineMetricMatching2012a}
Gupta, A., Lewi, K.: The {{Online Metric Matching Problem}} for {{Doubling
  Metrics}}. In: {EATCS} {International} {Colloquium} on {Automata},
  {{Languages}} and {{Programming}}. LNCS, vol.~7391, pp. 424--435. {Springer}
  (2012)

\bibitem{karpOptimalAlgorithmOnline1990a}
Karp, R.M., Vazirani, U.V., Vazirani, V.V.: An optimal algorithm for on-line
  bipartite matching. In: Proceedings of the Twenty-Second Annual {{ACM}}
  Symposium on {{Theory}} of {{Computing}}. pp. 352--358. {{STOC}} '90, {ACM}
  (1990)

\bibitem{khanPossibitityNtopologicalSpaces2012}
Khan, K.A.: On the possibitity of n-topological spaces. International Journal
  of Mathematical Archive  \textbf{3}(7),  2520--2523 (2012)

\bibitem{koutsoupiasOnlineMatchingProblem2004}
Koutsoupias, E., Nanavati, A.: The {{Online Matching Problem}} on a {{Line}}.
  In: First {Workshop} on {Approximation} and {Online} {Algorithms} ({WAOA}
  2003). LNCS, vol.~2909, pp. 179--191. {Springer} (2004)

\bibitem{liuImpatientOnlineMatching2018a}
Liu, X., Pan, Z., Wang, Y., Wattenhofer, R.: Impatient online matching. In:
  29th International Symposium on Algorithms and Computation, {ISAAC} 2018.
  LIPIcs, vol.~123, pp. 62:1--62:12. Schloss Dagstuhl - Leibniz-Zentrum
  f{\"{u}}r Informatik (2018)

\bibitem{liuOnlineMatchingConvex2022}
Liu, X., Pan, Z., Wang, Y., Wattenhofer, R.: Online matching with convex delay
  costs (2022), arXiv:2203.03335

\bibitem{mehtaAdWordsGeneralizedOnline2005a}
Mehta, A., Saberi, A., Vazirani, U., Vazirani, V.: {{AdWords}} and generalized
  on-line matching. In: 46th {{Annual IEEE Symposium}} on {{Foundations}} of
  {{Computer Science}} ({{FOCS}}'05). pp. 264--273 (2005)

\bibitem{mehtaOnlineMatchingAd2013a}
Mehta, A.: Online {{Matching}} and {{Ad Allocation}}. Foundations and Trends in
  Theoretical Computer Science  \textbf{8 (4)},  265--368 (2013)

\bibitem{melnyk_online_2021}
Melnyk, D., Wang, Y., Wattenhofer, R.: Online $k$-{{Way Matching}} with
  {{Delays}} and the {{$H$-Metric}} (2021), arXiv:2109.06640

\bibitem{nayyarInputSensitiveOnline2017a}
Nayyar, K., Raghvendra, S.: An {{Input Sensitive Online Algorithm}} for the
  {{Metric Bipartite Matching Problem}}. In: 2017 {{IEEE}} 58th {{Annual
  Symposium}} on {{Foundations}} of {{Computer Science}} ({{FOCS}} '17). pp.
  505--515 (2017)

\bibitem{pavoneOnlineHypergraphMatching2020}
Pavone, M., Saberi, A., Schiffer, M., Tsao, M.W.: {Technical Note―Online
  Hypergraph Matching with Delays}. Operations Research  \textbf{70}(4),
  2194--2212 (2022)

\bibitem{raghvendraRobustOptimalOnline2016}
Raghvendra, S.: {A Robust and Optimal Online Algorithm for Minimum Metric
  Bipartite Matching}. In: International Workshop on Approximation,
  Randomization, and Combinatorial Optimization. Algorithms and Techniques
  (APPROX/RANDOM 2016). LIPIcs, vol.~60, pp. 18:1--18:16. Schloss
  Dagstuhl--Leibniz-Zentrum fuer Informatik, Dagstuhl, Germany (2016)

\end{thebibliography}


%




\end{document}